\documentclass[aps,prl,twocolumn,superscriptaddress]{revtex4}

\usepackage{graphics}
\usepackage{epsfig}
\usepackage{amsmath}
\usepackage{amssymb}
\begin{document}
\newcommand{\hrho}{\widehat{\rho}}
\newcommand{\homega}{\widehat{\omega}}
\newcommand{\hI}{\widehat{I}}
\newcommand*{\spr}[2]{\langle #1 | #2 \rangle}
\newcommand*{\bbN}{\mathbb{N}}
\newcommand*{\bbR}{\mathbb{R}}
\newcommand*{\cB}{\mathcal{B}}
\newcommand*{\eps}{\varepsilon}
\newcommand*{\id}{I}
\newcommand{\orho}{\overline{\rho}}
\newcommand{\omu}{\overline{\mu}}
\newcommand*{\half}{{\frac{1}{2}}}
\newcommand*{\ket}[1]{| #1 \rangle}
\newcommand{\trho}{{\widetilde{\rho}}_n^\gamma}
\newcommand*{\bra}[1]{\langle #1 |}
\newcommand*{\proj}[1]{\ket{#1}\bra{#1}}
\newcommand{\otrho}{{\widetilde{\rho}}_n^{\gamma 0}}
\newcommand{\be}{\begin{equation}}
\newcommand{\bea}{\begin{eqnarray}}
\newcommand{\eea}{\end{eqnarray}}
\newcommand{\tr}{\mathrm{Tr}}
\newcommand*{\Hmin}{H_{\min}}
\newcommand{\rank}{\mathrm{rank}}
\newcommand{\tends}{\rightarrow}
\newcommand{\uS}{\underline{S}}
\newcommand{\oS}{\overline{S}}
\newcommand{\ee}{\end{equation}}
\newcommand{\n}{{(n)}}
\newtheorem{definition}{Definition}
\newtheorem{theorem}{Theorem}
\newtheorem{proposition}{Proposition}
\newtheorem{lemma}{Lemma}
\newtheorem{defn}{Definition}
\newtheorem{corollary}{Corollary}
\newcommand{\qed}{\hspace*{\fill}\rule{2.5mm}{2.5mm}}

\newenvironment{proof}{\noindent{\it Proof}\hspace*{1ex}}{\qed\medskip}
\def\reff#1{(\ref{#1})}
%


\title{Smooth R\'enyi Entropies and the Quantum Information Spectrum}


\author{Nilanjana Datta}
\email{N.Datta@statslab.cam.ac.uk}
\affiliation{Statistical Laboratory, DPMMS, University of Cambridge, Cambridge CB3 0WB, UK}
\author{Renato Renner}
\email{renner@phys.ethz.ch}
\affiliation{Institute for Theoretical Physics, ETH Zurich, 8093 Zurich, Switzerland}

\date{\today}

\begin{abstract}
Many of the traditional results in information theory, such as the
channel coding theorem or the source coding theorem, are restricted to
scenarios where the underlying resources are \emph{independent and
identically distributed (i.i.d.)} over a large number of uses. To overcome this
limitation, two different techniques, the \emph{information spectrum
  method} and the \emph{smooth entropy framework}, have been developed
independently. They are based on new entropy measures, called
\emph{spectral entropy rates} and \emph{smooth entropies},
respectively, that generalize Shannon entropy (in the classical case)
and von Neumann entropy (in the more general quantum case). Here, we
show that the two techniques are closely related. More precisely, the
spectral entropy rate can be seen as the asymptotic limit of the
smooth entropy. Our results apply to the quantum setting and thus
include the classical setting as a special case.
\end{abstract}


\maketitle

\section{Introduction}

Traditional results in information theory, e.g., the \emph{noisy
  channel coding theorem} or the \emph{source coding} (or \emph{data
  compression}) \emph{theorem}, typically rely on the assumption that
underlying resources, e.g., information sources and communication
channels, are ``memoryless''. A memoryless information source is one
which emits signals that are independent of each other.  Similarly, a
channel is said to be memoryless if the noise acting on successive
inputs to the channel is uncorrelated.  Such resources can be
described by a sequence of \emph{identical and independently
  distributed (i.i.d.)} random variables.

In reality, however, this
assumption cannot generally be justified. This is particularly
problematic in cryptography, where the accurate modeling of the
system is essential to derive any claim about its security.

In the past decade, two approaches have been proposed independently to
overcome this limitation. The \emph{information spectrum approach} was
introduced by Han and Verd\'u~\cite{hanverdu93,verdu94,han} in an
attempt to generalize the noisy channel coding theorem.  This approach
yields a unifying mathematical framework for obtaining asymptotic rate
formulae for many different operational schemes in information theory,
such as data compression, data transmission, and hypothesis
testing. The power of this method lies in the fact that it does not
rely on the specific nature of the sources or channels involved in the
schemes.

The main ingredients of this method are new entropy-type measures,
called \emph{spectral entropy rates}, which are defined asymptotically
for sequences of probability distributions. They can be seen as
generalizations of the Shannon entropy, and also inherit many of its
properties, such as subadditivity, strong subadditivity, monotonicity,
and Araki-Lieb inequalities. They also satisfy chain rule
inequalities. Their main feature, however, is that they characterize
various other asymptotic information-theoretic quantities, e.g., the
\emph{data compression rate}, without relying on the i.i.d.\
assumption.

Subsequently, Hayashi, Nagaoka, and Ogawa have generalized the
information-spectrum method to quantum-mechanical settings. They have
applied the method to study quantum hypothesis testing and quantum
source coding \cite{ogawa00, nagaoka02}, as well as to determine
general expressions for the optimal rate of entanglement concentration
\cite{hay_conc} and the classical capacity of quantum channels
\cite{hayashi03}. The method has been further extended by Bowen and
Datta~\cite{bd1} and used to obtain general formulae for the optimal
rates of various information-theoretic protocols, e.g., the dense
coding capacity for a noiseless quantum channel, assisted by arbitrary
shared entanglement \cite{bd_rev} and the entanglement cost for
arbitrary sequences of pure \cite{bd_entpure} and mixed
\cite{bd_entmixed} states. Recently, Matsumoto~\cite{Mat07} has also
employed the information spectrum method to obtain an alternative (but
equivalent) expression for the entanglement cost for an arbitrary
sequence of states.

In a simultaneous but independent development, the necessity to
generalize Shannon's theory became apparent in the context of
\emph{cryptography}.  Roughly speaking, one of the main challenges in
cryptography is that one needs to deal with an adversary who might
pursue an \emph{arbitrary} (and unknown) strategy. In particular, the
adversary might introduce undesired correlations which, for instance,
make it difficult to justify assumptions on the independence of noise
in a communication channel.

Bennett, Brassard, Cr\'epeau, and Maurer~\cite{BBCM95} were among the
first to make this point explicit, arguing that the Shannon entropy is
not an appropriate measure for the ignorance of an adversary about a
(partially secret) key. They proposed an alternative measure based on
the collision entropy (i.e., R\'enyi entropy~\cite{Renyi61} of order
$2$) and a notion called \emph{spoiling knowledge}, which can be seen
as a predecessor of smooth entropies.  This approach has been further
investigated by Cachin~\cite{Cachin97}, who also found connections to
other entropy measures, in particular \emph{R\'enyi entropies} of
arbitrary order.

Motivated by the work of Bennett \emph{et al.} and Cachin,
\emph{smooth R\'enyi entropies} have been introduced by Renner et al.,
first for the purely classical case (in~\cite{RenWol04b}), and later
for the more general quantum regime (in~\cite{KR,renatophd}).  In
contrast to the spectral entropy rates, smooth R\'enyi entropies are
defined for single distributions (rather than sequences of
distributions). Because of their \emph{non-asymptotic} nature, they
depend on an additional parameter $\eps$, called \emph{smoothness}.

Similarly to the spectral entropy rates, it has been shown that smooth
entropies have many properties in common with Shannon and von Neumann
entropy (for example, there is a chain rule, and strong subadditivity
holds)~\cite{RenWol05b,renatophd}.  Furthermore, they allow for a
quantitative analysis of a broad variety of information-theoretic
tasks|but in contrast to Shannon entropy, neither the i.i.d.\
assumption nor asymptotics are needed.  For example, in the classical
regime, it is possible to give a fully general formula for the number
of classical bits that can be transmitted reliably (up to some error
$\eps$) in \emph{one} (or finitely many) uses of a classical
channel~\cite{ReWoWu07}. In the quantum regime, they proved very
useful in the context of randomness extraction~\cite{KR,renatophd},
which, in turn, is used for cryptographic
applications~\cite{DFSS05,DFRSS07,DFSS07,STTV07}. In particular, they
are employed for the study of real-world implementations of
cryptographic schemes, where the available resources (e.g., the
computational power or the memory size) are finite~\cite{ScaRen07}.

%


Our aim in this paper is to find connections between the two different
approaches described above, by exploring the relationships between
spectral entropy rates and smooth entropies.  We do this in two
steps. First, we consider the special case where the entropies are not
conditioned on an additional system, in the following called the
\emph{non-conditional case}. Then, in a second step, we consider the
general \emph{conditional} case where the entropies are conditioned on
an extra system.

\section{Definitions of smooth entropy and spectral entropy rates}

\subsection{Mathematical Preliminaries}

Let ${\cal{B}}({\cal{H}})$ denote the algebra of linear operators
acting on a finite-dimensional Hilbert space ${\cal{H}}$. The von
Neumann entropy of a state $\rho$, i.e., a positive operator of unit
trace in ${\cal{B}}({\cal{H}})$, is given by $S(\rho) = - \tr \rho
\log \rho$. Throughout this paper, we take the logarithm to base $2$
and all Hilbert spaces considered are finite-dimensional.

The quantum information spectrum approach requires the extensive use of spectral projections. Any self-adjoint operator $A$ acting on a finite-dimensional Hilbert space may be written in its spectral
decomposition $A = \sum_i \lambda_i |i\rangle \langle i|$.  We define the
positive spectral projection on $A$ as $\{ A \geq 0 \} := \sum_{\lambda_i \geq 0} |i\rangle \langle i|$, the projector onto the eigenspace of $A$ corresponding to positive eigenvalues.  Corresponding definitions apply for the
other spectral projections $\{ A < 0 \}, \{ A > 0 \}$ and $\{ A \leq 0 \}$. For two operators $A$ and $B$, we can
then define $\{ A \geq B \}$ as $\{ A - B \geq 0 \}$.  The following key lemmas are useful. For a proof of Lemma \ref{lem1}, see \cite{ogawa00,nagaoka02}.
\begin{lemma}
\label{lem1}
For self-adjoint operators $A$, $B$ and any positive operator $0 \leq P \leq I$
the inequality we have
\bea
\mathrm{Tr}\big[ P(A-B)\big] &\leq& \mathrm{Tr}\big[ \big\{ A \geq B \big\}
(A-B)\big]\label{lem11}\\
\mathrm{Tr}\big[ P(A-B)\big] &\geq& \mathrm{Tr}\big[ \big\{ A \leq B \big\}
(A-B)\big].
\label{lem12}
\eea
Identical conditions hold for strict inequalities in the spectral
projections $\{A < B\}$ and $\{ A > B\}$.
\end{lemma}
\begin{lemma}
\label{lem2}
Given a state $\rho_n$ and a self-adjoint
operator $\omega_n$, for any real $\gamma$ we have
$$
\mathrm{Tr}\big[\{\rho_n \ge 2^{-n\gamma}\omega_n \} \omega_n \bigr]
\leq 2^{n\gamma}.
$$
\end{lemma}
\begin{proof}
Note that
$$
\mathrm{Tr}\big[\{\rho_n \ge 2^{-n\gamma}\omega_n \}
(\rho_n-2^{-n\gamma}\omega_n) \bigr]\ge 0
$$
Hence,
\bea
2^{-n\gamma}\mathrm{Tr}\big[\{\rho_n \ge 2^{-n\gamma}\omega_n \}
\omega_n \bigr]&\le&\mathrm{Tr}\big[\{\rho_n \ge 2^{-n\gamma}\omega_n \}
\rho_n \bigr] \nonumber\\
&\le& \mathrm{Tr} \rho_n = 1
\eea
Therefore,
$$
\mathrm{Tr}\big[\{\rho_n \ge 2^{-n\gamma}\omega_n \} \omega_n \bigr]
\leq 2^{n\gamma}.
$$
\end{proof}

The trace distance between two operators $A$ and $B$ is given by
\be
||A-B||_1 := \tr\bigl[\{A \ge B\}(A-B)\bigr] -
 \tr\bigl[\{A < B\}(A-B)\bigr]
\ee
The fidelity of states $\rho$ and $\rho'$ is defined to be
$$ F(\rho, \rho'):= \tr \sqrt{\rho^{\half} \rho' \rho^{\half}}.
$$
The trace distance between two states $\rho$ and $\rho'$ is
related to the fidelity $ F(\rho, \rho')$ as follows (see (9.110) of \cite{nielsen}):
\be
  \frac{1}{2} \| \rho - \rho' \|_1
\leq
  \sqrt{1-F(\rho, \rho')^2}
\leq
  \sqrt{2(1-F(\rho, \rho'))} \ .
\label{fidelity}
\ee
We also use the following simple corollary of Lemma \ref{lem1}:
\begin{corollary}
\label{cor1}
For self-adjoint operators $A$, $B$ and any positive operator $0 \leq P \leq
I$,
the inequality
$$||A-B||_1 \le \eps,$$
for any $\eps >0$,
implies that $$\tr \bigl[P(A-B) \bigr] \le \eps.$$
\end{corollary}
We also use the ``gentle measurement'' lemma \cite{winter99,ogawanagaoka02}.
\begin{lemma}\label{gm} For a state $\rho$ and operator $0\le \Lambda\le I$, if
$\mathrm{Tr}(\rho \Lambda) \ge 1 - \delta$, then
$$||\rho -   {\sqrt{\Lambda}}\rho{\sqrt{\Lambda}}||_1 \le {2\sqrt{\delta}}.$$
The same holds if $\rho$ is only a subnormalized density operator.
\end{lemma}

\subsection{Definition of spectral divergence rates}
In the quantum information spectrum approach one defines spectral divergence
rates, defined below, which can be viewed as generalizations of the quantum relative entropy.
\begin{definition}
Given a sequence of states $\hrho=\{\rho_n\}_{n=1}^\infty$ and a sequence of positive operators
$\homega=\{\omega_n\}_{n=1}^\infty$,
the quantum spectral sup-(inf-)divergence rates are defined in terms
of the difference operators $\Pi_n(\gamma) = \rho_n - 2^{n\gamma}\omega_n$ as
\begin{align}
\overline{D}(\hrho \| \homega) &:= \inf \Big\{ \gamma : \limsup_{n\rightarrow \infty} \mathrm{Tr}\big[ \{ \Pi_n(\gamma) \geq 0 \} \Pi_n(\gamma) \big] = 0 \Big\} \label{od} \\
\underline{D}(\hrho \| \homega) &:= \sup \Big\{ \gamma : \liminf_{n\rightarrow \infty} \mathrm{Tr}\big[ \{ \Pi_n(\gamma) \geq 0 \}
\Pi_n(\gamma) \big] = 1 \Big\} \label{ud}
\end{align}
respectively.
\end{definition}
Although the use of sequences of states
allows for immense freedom in choosing them,
there remain a number of basic properties of the quantum spectral divergence
rates that hold for all sequences. These are stated and proved in
\cite{bd1}.  In the {i.i.d.} case the sequence is
generated from product states $\rho = \{ \varrho^{\otimes n}
\}_{n=1}^{\infty}$, which is used to relate the spectral entropy rates for the
sequence $\rho$ to the entropy of a single state $\varrho$.

Note that the above definitions of the spectral divergence rates
differ slightly from those originally given in (38) and (39) of
\cite{hayashi03}. However, they are equivalent, as stated in the
following two propositions (proved in \cite{bd1}).  The proofs have
been included in the Appendix for completeness.
\begin{proposition}
\label{equiv_sup}
The spectral sup-divergence rate $\overline{D}(\rho\| \omega)$ is equal to
\begin{equation}
\overline{\mathcal{D}}(\rho\| \omega) = \inf \Big\{ \alpha : \limsup_{n\rightarrow
\infty} \mathrm{Tr}\big[ \{ \rho_n \geq e^{n\alpha}\omega_n \} \rho_n \big] = 0
\Big\}
\end{equation}
which is the previously used definition of the spectral sup-divergence rate.
Hence the two definitions are equivalent.
\end{proposition}

\begin{proposition}
\label{equiv_inf}
The spectral inf-divergence rate $\underline{D}(\rho\| \omega)$ is equivalent
to
\begin{equation}
\underline{\mathcal{D}}(\rho\| \omega) = \sup \Big\{ \alpha :
\liminf_{n\rightarrow \infty} \mathrm{Tr}\big[ \{ \rho_n \geq e^{n\alpha}\omega_n
\} \rho_n \big] = 1 \Big\}
\end{equation}
which is the previously used definition of the spectral inf-divergence rate.
\end{proposition}

Despite these equivalences, it is useful to use the definitions
\reff{od} and \reff{ud} for the divergence rates as they
allow the application of Lemmas \ref{lem1} and \ref{lem2} in deriving
various properties of these rates.

The spectral generalizations of the von Neumann entropy, the conditional
entropy and the mutual information can all be expressed as spectral divergence rates with appropriate
substitutions for the sequence of operators $\homega = \{ \omega_n
\}_{n=1}^{\infty}$.

\subsection{Definition of spectral entropy rates}

Consider a sequence of Hilbert spaces $\{{\cal{H}}_n\}_{n=1}^\infty$, with
${\cal{H}}_n = {\cal{H}}^{\otimes n}$.
For any sequence of states $\hrho=\{\rho_n\}_{n=1}^\infty$, with $\rho_n$ being a density matrix acting in the
Hilbert space ${\cal{H}}_n$, the sup- and inf- spectral entropy rates are defined as follows:
\begin{align}
\overline{S}(\hrho) &= \inf \Big\{ \gamma : \liminf_{n\rightarrow \infty} \mathrm{Tr}\big[ \{ \rho_n \geq 2^{-n\gamma} I_n\} \rho_n \big] = 1 \Big\} \label{os} \\
\underline{S}(\hrho) &= \sup \Big\{ \gamma : \limsup_{n\rightarrow \infty} \mathrm{Tr}\big[ \{ \rho_n \geq 2^{-n\gamma} I_n\} \rho_n\big] = 0 \Big\}.
\label{us}
\end{align}
Here $I_n$ denotes the identity operator acting in ${\cal{H}}_n$.
These are obtainable from the spectral divergence rates as follows [see \cite{bd1}:
\be
\overline{S}(\hrho)= - \underline{D} (\hrho|| \widehat{I}) \,;\,\underline{S}(\hrho)= - \overline{D} (\hrho|| \widehat{I}),
\label{spec}\ee
where $\widehat{I} = \{I_n\}_{n=1}^\infty$ is a sequence of identity operators.

It is known \cite{bd1} that the spectral entropy rates of $\hrho$ are related to the von Neumann entropies of the states $\rho_n$
as follows:
\be
\underline{S}(\hrho) \le \liminf_{n\rightarrow \infty} \frac{1}{n} S(\rho_n) \le \limsup_{n\rightarrow \infty} \frac{1}{n} S(\rho_n) \le \overline{S}(\hrho).
\ee

Moreover for a sequence of states $\hrho=\{\rho^{\otimes n}\}_{n=1}^\infty$:
\be
\underline{S}(\hrho)= \lim_{n\rightarrow \infty} \frac{1}{n} S(\rho_n)= \overline{S}(\hrho).
\ee


For sequences of bipartite states $\hrho = \{\rho_n^{AB}\}_{n=1}^\infty$,
with $\rho_n^{AB} \in {\cal{B}}\left(({\cal{H}}_A \otimes
{\cal{H}}_B)^{\otimes n}\right)$, the conditional spectral entropy rates
are defined as follows:
\bea
\overline{S}(A|B) &:=& -\underline{D}(\hrho^{AB}| \hI^{A}\otimes \hrho^B)
\label{ocond};\\
\underline{S}(A|B) &:=& -\overline{D}(\hrho^{AB}| \hI^{A}\otimes \hrho^B)
.
\label{ucond}
\eea
In the above,
$\hI^{A}=\{I^A_n\}_{n=1}^\infty$ and $\hrho^{A}=\{\rho^A_n\}_{n=1}^\infty$,
with $I^A_n$ being the identity operator acting in in
${\cal{H}}_A^{\otimes n}$
and $\rho^A_n = \mathrm{Tr}_B \rho^{AB}_n$, the partial trace
being taken on the Hilbert space ${\cal{H}}_B^{\otimes n}$.

\subsection{Definition of  min- and max-entropies}

We start with the definition of \emph{non-smooth} min- and max-entropies.

\begin{definition}[\cite{renatophd}]
  The \emph{min-} and \emph{max-entropies} of a bipartite state
  $\rho_{A B}$ relative to a state $\sigma_B$ are defined by
  \[
    H_{\min}(\rho_{A B} | \sigma_B)
  :=
    - \log \min\{ \lambda: \, \rho_{A B} \leq \lambda \cdot \id_A \otimes \rho_B \}
  \]
  and
  \[
    H_{\max}( \rho_{A B} | \sigma_B )
  :=
    \log \tr\bigl(\pi_{A B} (\id_A \otimes \sigma_B)\bigr) \ ,
  \]
  where $\pi_{A B}$ denotes the projector onto the support of $\rho_{A B}$.
\end{definition}

In the special case where the system $B$ is trivial (i.e.,
$1$-dimensional), we simply write $H_{\min}(\rho_A)$ and
$H_{\max}(\rho_A)$. These entropies then correspond to the usual
\emph{non-conditional} R\'enyi entropies of order infinity and zero,
\begin{align*}
  H_{\min}(\rho_A) = H_{\infty}(\rho_A) & = - \log \| \rho_A \|_{\infty} \\
  H_{\max}(\rho_A) = H_0(\rho_A) & = \log \rank(\rho_A) \ ,
\end{align*}
where $\| \cdot \|_{\infty}$ denotes the $L_{\infty}$-norm.

\subsection{Definition of smooth min- and max-entropies}

\emph{Smooth} min- and max-entropies are generalizations of the above
entropy measures, involving an additional \emph{smoothness} parameter
$\eps \geq 0$. For $\eps = 0$, they reduce to the
\emph{non-smooth} quantities.

\begin{definition}[\cite{renatophd}] \label{def:smoothentropies} For
  any $\eps \geq 0$, the \emph{$\eps$-smooth min-} and
  \emph{max-entropies} of a bipartite state $\rho_{A B}$ relative to a
  state $\sigma_B$ are defined by
\[
    H_{\min}^{\eps}(\rho_{A B} | \sigma_B)
  :=
    \sup_{\bar{\rho} \in B^{\eps}} H_{\min}(\bar{\rho} | \sigma_B)
  \]
  and
  \[
    H_{\max}^{\eps}( \rho_{A B} | \sigma_B )
  :=
    \inf_{\bar{\rho} \in B^{\eps}} H_{\max}(\bar{\rho} | \sigma_B)
  \]
  where $B^{\eps}(\rho) := \{\bar{\rho} \geq 0: \, \| \bar{\rho} - \rho
  \|_1 \leq \eps, \tr(\bar{\rho}) \leq \tr(\rho)\}$.
\end{definition}

In the following, we will focus on the smooth min- and max-entropies
for the case where $\sigma_B = \rho_B$. Note that the
quantities~$H_{\min}^\eps(\rho_{A B}|B) := \max_{\sigma_B}
H_{\min}^\eps(\rho_{A B} | \sigma_B)$ and $H_{\max}^\eps(\rho_{A B}|B) :=
\min_{\sigma_B} H_{\max}^{\eps}(\rho_{A B} | \sigma_B)$ defined in~\cite{renatophd}
are not studied in this paper.

\section{Relation between non-conditional entropies}

\subsection{Relation between $\underline{S}(\hrho)$ and $H_{\min}^\eps(\rho)$}

\begin{theorem} \label{thm:relSHmin}

Given a sequence of states $\hrho=\{\rho_n\}_{n=1}^\infty$, where $\rho_n\in {\cal{B}}({\cal{H}}_{n})$,
with ${\cal{H}}_{n} = {\cal{H}}^{\otimes {n}}$, the inf-spectral entropy rate $\uS(\hrho)$
is related to the smooth min-entropy as follows:
\be
\uS(\hrho)= \lim_{\eps \rightarrow 0} \liminf_{n \rightarrow \infty}
\frac{1}{n} H^\eps_{\min}(\rho_n)
\label{main}
\ee
\end{theorem}

\begin{proof}
For any constant $\gamma>0$, let us define projection operators
\be Q_n^\gamma := \{\rho_n< 2^{-n\gamma} I_n\}\,
\label{proj1}\ee
and \be
 P_n^\gamma := I_n -  Q_n^\gamma = \{\rho_n\ge 2^{-n\gamma} I_n\}.\\
\label{proj2}
\ee
In terms of these projections, we can write
\be
\underline{S}(\hrho) = \sup \Big\{ \gamma : \limsup_{n\rightarrow \infty} \mathrm{Tr}\big[ P_n^\gamma
\rho_n\big] = 0 \Big\},
\label{us1}
\ee
or alternatively as
\be
\underline{S}(\hrho) = \sup \Big\{ \gamma : \liminf_{n\rightarrow \infty} \mathrm{Tr}\big[ Q_n^\gamma
\rho_n\big] = 1 \Big\},
\label{uss}
\ee
since each $\rho_n$ in the sequence $\hrho$ is a state (i.e., $\tr \rho_n = 1$).
From Proposition \ref{equiv_inf} and \reff{spec} of $\underline{S}(\hrho)$
it follows that the latter is equivalently
given by the expression
\be
\underline{S}(\hrho) = \sup \Big\{ \gamma : \limsup_{n\rightarrow \infty} \mathrm{Tr}\big[ P_n^\gamma
(\rho_n- 2^{-n\gamma} I_n)\big] = 0 \Big\},
\label{us2}
\ee

From \reff{uss} it follows that, for any $\gamma < \uS(\hrho)$ and any $\delta >0$, for $n$ large enough,
\be\tr \big[Q_n^\gamma
\rho_n\big] > 1 - \delta.\label{one} \ee

For {{any}} given $\alpha >0$, let $\gamma:= \uS(\hrho) - \alpha$, and let
\be \trho := Q_n^\gamma  \rho_n Q_n^\gamma
\label{trho}
\ee Then using (\ref{one})
and Lemma \ref{gm} we infer that, for $n$ large enough,
\be
||  \rho_n - \trho||_1 \le 2\sqrt{\delta}.
\label{onen}
\ee
In other words, for $n$ large enough, $\trho \in B^\eps(\rho_n)$ with $\eps = 2{\sqrt{\delta}}$.

We first prove the upper bound
\be
\uS(\hrho) \le \lim_{\eps \rightarrow 0} \liminf_{n \rightarrow \infty}
\frac{1}{n} H^\eps_{\min}(\rho_n)
\label{first}
\ee
For $n$ large enough,
\bea
H_{\min}^\eps( \rho_n)
&\equiv& \sup_{\orho_n\in B^\eps(\rho_n)}
H_{\min}(\orho_n)\nonumber\\
&\ge & H_{\min}(\trho)= - \log \| \trho \|_{\infty} \nonumber\\
&>& n \gamma = n(\uS(\hrho) - \alpha) \label{direct}
\eea
The last line follows from the inequality $\trho < 2^{-n\gamma} I_n$, and
since $\alpha$ is arbitrary,
we obtain the desired bound \reff{first}.

We next prove the converse, i.e.,
\be
\uS(\hrho) \ge \lim_{\eps \rightarrow 0} \liminf_{n \rightarrow \infty}
\frac{1}{n} H^\eps_{\min}(\rho_n)
\label{second}
\ee
Consider an operator $\orho_n^\eps \in B^\eps (\rho_n)$ for which
\be
-\log \| \orho_n^\eps \|_{\infty} = \sup_{\orho_n \in B^\eps (\rho_n)}
\bigl[ -\log \| \orho_n \|_{\infty} \bigr].
\label{orho}
\ee

We shall also make use of a quantity $\underline{\Upsilon}(\homega)$,
defined for any sequence of positive operators $\homega = \{ \omega_n
\}_{n=1}^{\infty}$ as follows:
\be
\underline{\Upsilon}(\homega) = \sup \Big\{ \alpha : \limsup_{n\rightarrow \infty} \mathrm{Tr}\big[
\{\omega_n \ge 2^{-n\alpha} I_n\}
\Pi_n^\alpha\big] = 0 \Big\},
\label{us22}
\ee where $\Pi_n^\alpha:= (\omega_n- 2^{-n\alpha} I_n)$.  Note that
$\underline{\Upsilon}(\homega)$ reduces to the inf-spectral entropy
rate $\underline{{S}}(\homega)$ given by \reff{us2}, if $\homega$ is a
sequence of states.

By the definition of the smooth min-entropy, (\ref{second}) then follows
from Lemma~\ref{above} below.

\end{proof}

\begin{lemma}
\label{above}
For any sequence of states $\hrho = \{\rho_n\}_{n=1}^\infty$, and any
$\eps >0$, there exists an $n_0 \in \mathbb{N}$, such that for all $n
\ge n_0$ \be \uS(\hrho) \ge \frac{1}{n} \bigl[ - \log \| \orho_n^\eps
\|_{\infty} \bigr],
\label{last}
\ee
with $\orho_n^\eps$ defined by \reff{orho}.
\end{lemma}

\begin{proof}
We prove this in two steps. We first prove that for any $\eps >0$ and $n$ large enough,
\be
\underline{\Upsilon}(\hrho^\eps) \ge-\frac{1}{n} \log \| \orho_n^\eps \|_{\infty},
\label{step1}
\ee
where $\hrho^\eps := \{\orho_n^\eps\}_{n=1}^\infty$.
We then prove that
\be
\lim_{\eps \rightarrow 0} \underline{\Upsilon}(\hrho^\eps) \le
\uS(\hrho)
\label{step2}
\ee

For any arbitrary $\eta >0$, let $\alpha$ be defined through
the relation
\be \|\orho_n^\eps \|_{\infty} = 2^{-n(\alpha + \eta)}.
\label{alpha}
\ee
This implies the operator inequality, $\orho_n^\eps- 2^{-n(\alpha + \eta)}I_n \le 0$,
and hence
$\orho_n^\eps < 2^{-n\alpha}I_n$.

Hence,
\be\tr\bigl[\{\orho_n^\eps \ge 2^{-n\alpha} I_n\}
(\orho_n^\eps - 2^{-n\alpha} I_n) ] = 0,
\label{two}
\ee
Using this, and the
definition of $\underline{\Upsilon}(\hrho^\eps)$, we infer
that $\alpha \le \underline{\Upsilon}(\hrho^\eps)$. Then, using
\reff{alpha} we obtain the bound
$$
- \frac{1}{n} \log \| \orho_n^\eps \|_{\infty} - \eta \le \underline{\Upsilon}(\hrho^\eps),$$
which in turn yields \reff{step1}, since $\eta$ is arbitrary.

To prove \reff{step2} note that
\bea
0 &\le& \tr(P_n^\gamma \rho_n)\nonumber\\
&=& \tr(P_n^\gamma \orho_n^\eps) + \tr\bigl[P_n^\gamma (\rho_n - \orho_n^\eps)\bigr]\nonumber\\
&\le & \tr\bigl[P_n^\gamma(\orho_n^\eps - 2^{-n\alpha} I_n)\bigr] + 2^{-n\alpha} \tr P_n^\gamma + \eps\nonumber\\
&\le& \tr\bigl[ \{\orho_n^\eps \ge 2^{-n\alpha} I_n\} (\orho_n^\eps - 2^{-n\alpha} I_n)\bigr]
+ 2^{-n(\alpha - \gamma)} + \eps.\nonumber\\
\label{eq1}
\eea

The third line in \reff{eq1} is obtained by using the bound
$$\tr\bigl[ P_n^\gamma(\rho_n - \orho_n^\eps)\bigr]\le \eps, $$
which follows from Corollary \ref{cor1}, since $\orho_n^\eps
\in B^\eps(\rho_n)$.

To arrive at the last line of \reff{eq1} we use Lemma \ref{lem1} and
the fact that $\tr{P_n^\gamma} \le 2^{n\gamma}$, which follows from
Lemma \ref{lem2}.

Let us choose $\gamma = \alpha - \delta/2$, for an arbitrary
$\delta>0$, with $\alpha= \underline{\Upsilon}(\hrho^\eps) -
\delta/2$. Then both the first and second terms on the r.h.s.\ of
\reff{eq1} goes to zero as $n \tends \infty$. Therefore, for $n$ large
enough and any $\delta^{'}>0$, in the limit $\eps \tends 0$, we must
have that \be\tr(P_n^\gamma \rho_n) \le \delta^{'},\ee which in turn
implies that $\gamma \le \uS(\hrho)$.

From the choice of the parameters $\alpha$ and $\gamma$ it follows
that \be \lim_{\eps \tends 0} \underline{{\Upsilon}}(\hrho^\eps ) -
\delta < \underline{{{S}}}(\hrho).\ee But since $\delta$ is arbitrary,
we obtain the inequality \reff{step2}.
\end{proof}

\subsection{Relation between $\overline{S}(\hrho)$ and $H_{\max}^\eps(\rho)$}
\begin{theorem}
\label{max}
Given a sequence of states $\hrho=\{\rho_n\}_{n=1}^\infty$, where $\rho_n\in {\cal{B}}({\cal{H}}_{n})$,
with ${\cal{H}}_{n} = {\cal{H}}^{\otimes {n}}$, the sup-spectral entropy rate $\oS(\hrho)$
is related to the smooth max-entropy as follows:
\be
\oS(\hrho)= \lim_{\eps \rightarrow 0} \limsup_{n \rightarrow \infty}
\frac{1}{n} H^\eps_{\max}(\rho_n)
\label{main3}
\ee
\end{theorem}

\begin{proof}
By definition, the sup-spectral entropy rate for the given sequence of states is
\be
\oS(\hrho) = \inf \Big\{ \gamma : \liminf_{n\rightarrow \infty} \mathrm{Tr}\big[ P_n^\gamma
\rho_n\big] = 1 \Big\},
\label{os1}
\ee
where $P_n^\gamma$ is the projection operator defined by \reff{proj2}.

From \reff{os1} it follows that, for any $\gamma \ge \oS(\hrho)$ and any $\delta >0$, for $n$ large enough
\be\tr \big[P_n^\gamma
\rho_n\big] > 1 - \delta.\label{o1} \ee

For {{any}} given $\alpha >0$, choose $\gamma= \oS(\hrho) + \alpha$, and let
\be \trho := P_n^\gamma  \rho_n P_n^\gamma
\label{otrho}
\ee Then using (\ref{o1})
and Lemma \ref{gm} we infer that, for $n$ large enough,
\be
||  \rho_n - \trho||_1 \le 2\sqrt{\delta}.
\label{on1n}
\ee
and hence $\trho \in B^\eps(\rho_n)$ with $\eps = 2{\sqrt{\delta}}$.

We first prove the bound
\be
\lim_{\eps \rightarrow 0} \limsup_{n \rightarrow \infty}
\frac{1}{n} H^\eps_{\max}(\rho_n)\le \oS(\hrho)
\label{ofirst}
\ee

For $n$ large enough,
\bea
H_{\max}^\eps( \rho_n)
&=& \inf_{\orho_n\in B^\eps(\rho_n)}
H_{\max}(\orho_n)\nonumber\\
&\le & H_{\max}(\trho)\nonumber\\
&=& \log \rank\, (\trho)
\label{oconverse}
\eea

From the definition \reff{otrho} of $\trho$ it follows that
$\rank\,{\trho} \le \tr P_n^\gamma$. Hence,
\bea
H_{\max}^\eps( \rho_n)
&\le& \log \tr P_n^\gamma \nonumber\\
&\le& n\gamma = \oS(\hrho) + \alpha,
\label{ocon1}
\eea
where once again we use the bound $\tr{P_n^\gamma} \le 2^{n\gamma}$.
The last line of \reff{ocon1} yields the desired bound \reff{ofirst}
since $\alpha$ is arbitrary.

To complete the proof of Theorem \ref{max} we assume that
\be
\lim_{\eps \rightarrow 0} \limsup_{n \rightarrow \infty}
\frac{1}{n} H^\eps_{\max}(\rho_n) < \oS(\hrho)
\label{osec}
\ee
and show that this leads to a contradiction.

Let $\sigma_{n,\eps}$ be the operator for which
\be
H_{\max} (\sigma_{n,\eps}) := \inf_{\orho_n \in B^\eps(\rho_n)}H_{\max} (\orho_n).\ee
Hence, $H^\eps_{\max}(\rho_n) = \log \rank\, \sigma_{n,\eps}$,
and the assumption \reff{osec} is equivalent to the following assumption:
\be
\lim_{\eps \rightarrow 0} \lim_{n \rightarrow \infty}
\frac{1}{n} \log \rank\, \sigma_{n,\eps} < \oS(\hrho).
\label{assume}
\ee
Since $\sigma_{n,\eps} \in B^\eps(\rho_n)$, $\tr \sigma_{n,\eps} \ge 1- \eps$. Let $\sigma_{n,\eps}^0$
denote the projection onto the support of $\sigma_{n,\eps}$. Then
\bea
\tr\bigl(\sigma_{n,\eps}^0\rho_n\bigr) &=& \tr \bigl[ \bigl((\rho_n - \sigma_{n,\eps})
+ \sigma_{n,\eps}\bigr)\sigma_{n,\eps}^0\bigr]\nonumber\\
&=& \tr \bigl[(\rho_n - \sigma_{n,\eps}) \sigma_{n,\eps}^0\bigr] + \tr \sigma_{n,\eps}\nonumber\\
&\ge& \tr \bigl[ \{\rho_n \le \sigma_{n,\eps}\} (\rho_n - \sigma_{n,\eps})\bigr] + 1- \eps\nonumber\\
&\ge & - \eps + 1 - \eps = 1- 2\eps.
\label{upb4}
\eea
The inequality in the third line follows from Lemma \ref{lem1}.
We arrive at the last inequality in \reff{upb4} by using the bound
$$\tr \bigl[ \{\rho_n \le \sigma_{n,\eps}\} (\rho_n - \sigma_{n,\eps})\bigr] \ge - \eps,$$
which arises from the fact that $\sigma_{n,\eps} \in B^\eps(\rho_n)$.

Note, however, that for $n$ large enough, \reff{upb4} leads to a contradiction, in the limit $\eps \rightarrow 0$. This is because, for any real number $R < \oS(\hrho)$
and any projection $\pi_n$, with $\tr \pi_n = 2^{nR}$, for $n$ large
enough, we have
\be
\tr(\pi_n \rho_n) \le 1 - c_0,
\label{compress}
\ee
for some constant $c_0 > 0$.
The inequality \reff{compress} can be proved as follows:

\bea
\tr(\pi_n \rho_n) &=& \tr\bigl[\pi_n (\rho_n - 2^{-n\beta} I_n)\bigr] + 2^{-n\beta} \tr \pi_n\nonumber\\
&\le & \tr\bigl[\{\rho_n \ge 2^{-n\beta} I_n\}(\rho_n - 2^{-n\beta} I_n) \bigr] \nonumber\\
& & + 2^{-n(\beta - R)}
\nonumber\\
\label{upb5}
\eea
Choose $\oS(\hrho) > \beta > R$. For such a choice, the second term on the right  hand side of
\reff{upb5} tends to zero asymptotically in $n$. However, the first term does not tend to $1$
and we hence obtain the bound \reff{compress}.

\end{proof}

\section{Relation between conditional entropies}

Consider a sequence of bipartite states
$\hrho^{AB}=\{\rho_n^{AB}\}_{n=1}^\infty$, with $\rho_n^{AB} \in
{\cal{B}}\bigl(({\cal{H}}_A \otimes {\cal{H}}_B)^{\otimes n}
\bigr)$. Let $\hrho^{AB}=\{\rho_n^{AB}\}_{n=1}^\infty$ denote the
corresponding sequence of reduced states.

For the sequence $\hrho^{AB}$, the sup-spectral conditional entropy
rate $\oS(A|B)$ and the inf-spectral conditional entropy rate
$\uS(A|B)$, defined respectively by \reff{ocond} and \reff{ucond}, can
be expressed as follows: \bea
\overline{S}(A|B) &=& \inf \Big\{ \gamma : \liminf_{n\rightarrow \infty} \mathrm{Tr}\big[P_n^\gamma \rho_n^{AB}\big] = 1 \Big\}, \label{oscond}\\
\underline{S}(A|B) &=& \sup \Big\{ \gamma : \limsup_{n\rightarrow
  \infty} \mathrm{Tr}\big[P_n^\gamma \rho_n^{AB}\big] = 0
\Big\}, \label{uscond} \eea where \be P_n^\gamma := \{ \rho_n^{AB}
\geq 2^{-n\gamma} I_n^A \otimes \rho_n^{B}\}.\label{proj21}\ee Here
$I_n^A$ denotes the identity operator in
${\cal{B}}({\cal{H}}_A^{\otimes n})$.

We use the following key properties of $H_{\min}^\eps
(\rho_{AB}|\rho_{B})$ given by Lemma \ref{lem:Hminepsaddbound} and
Lemma \ref{lem:Hminepsprojbound} below.

\begin{lemma} \label{lem:Hminepsaddbound}
  Let $\rho_{A B}$ and  $\sigma_B$ be density operators, let $\Delta_{A B}$ be a positive operator, and let $\lambda \in \bbR$  such that
   \[
    \rho_{A B} \leq 2^{-\lambda} \cdot \id_A \otimes \sigma_B + \Delta_{A B} \ .
  \]
  Then $\Hmin^{\eps}(\rho_{A B}|\sigma_B) \geq \lambda$ for any $\eps \geq \sqrt{8 \tr(\Delta_{A B})}$.
\end{lemma}

\begin{proof}
  Define
  \begin{align*}
  \alpha_{A B} & := 2^{-\lambda} \cdot \id_A \otimes \sigma_B \\
  \beta_{A B} & := 2^{-\lambda} \cdot \id_A \otimes \sigma_B + \Delta_{A B} \ .
  \end{align*}
 and
  \[
    T_{A B} := \alpha_{A B}^\half \beta_{A B}^{-\half} \ .
  \]
  Let $\ket{\Psi} = \ket{\Psi}_{A B R}$ be a purification of $\rho_{A B}$ and let $\ket{\Psi'} := T_{A B} \otimes \id_R \ket{\Psi}$ and $\rho'_{A B} := \tr_R(\proj{\Psi'})$.

Note that
\begin{align*}
  \rho'_{A B}
& =
  T_{A B} \rho_{A B} T_{A B}^{\dagger} \\
& \leq
  T_{A B} \beta_{A B} T_{A B}^{\dagger} \\
& =
  \alpha_{A B}
=
  2^{-\lambda} \cdot \id_A \otimes \sigma_B \ ,
\end{align*}
which implies $\Hmin(\rho'_{A B}|\sigma_B) \geq \lambda$. It thus remains to be shown that
\begin{equation} \label{eq:distbound}
  \| \rho_{A B} - \rho'_{A B} \|_1
\leq
  \sqrt{8 \tr(\Delta_{A B})}
 \ .
\end{equation}

  We first show that the Hermitian operator
  \[
    \bar{T}_{A B} := \frac{1}{2} (T_{A B} + T_{A B}^\dagger) \ .
  \]
  satisfies
  \begin{equation} \label{eq:Tleqid}
    \bar{T}_{A B} \leq \id_{A B} \ .
 \end{equation}
 For any vector $\ket{\phi} = \ket{\phi}_{A B}$,
 \begin{align*}
   \| T_{A B} \ket{\phi} \|^2
 & =
   \bra{\phi} T_{A B}^\dagger T_{A B} \ket{\phi}
=
   \bra{\phi} \beta_{A B}^{-\half} \alpha_{A B} \beta_{A B}^{-\half} \ket{\phi} \\
 & \leq
   \bra{\phi} \beta_{A B}^{-\half} \beta_{A B} \beta_{A B}^{-\half} \ket{\phi}
 =
   \| \ket{\phi} \|^2
  \end{align*}
 where the inequality follows from $\alpha_{A B} \leq \beta_{A B}$.
 Similarly,
 \begin{align*}
   \| T_{A B} ^\dagger \ket{\phi} \|^2
 & =
   \bra{\phi} T_{A B} T_{A B}^\dagger \ket{\phi}
 =
   \bra{\phi} \alpha_{A B}^\half \beta_{A B}^{-1} \alpha_{A B}^{\half} \ket{\phi} \\
 & \leq
   \bra{\phi} \alpha_{A B}^{\half} \alpha_{A B}^{-1} \alpha_{A B}^{\half} \ket{\phi}
 =
   \| \ket{\phi} \|^2
  \end{align*}
 where the inequality follows from the fact that $\beta_{A B}^{-1} \leq \alpha_{A B}^{-1}$ which holds because the function $\tau \mapsto -\tau^{-1}$ is operator monotone on $(0, \infty)$ (see Proposition V.1.6 of \cite{bhatia}). We conclude that for any vector $\ket{\phi}$,
 \begin{align*}
  \| \bar{T}_{A B} \ket{\phi} \|
& \leq
   \frac{1}{2} \| T_{A B} \ket{\phi} + T_{A B}^{\dagger} \ket{\phi} \| \\
 & \leq
   \frac{1}{2} \| T_{A B} \ket{\phi} \| + \frac{1}{2} \| T_{A B}^{\dagger} \ket{\phi} \|
 \leq
   \| \ket{\phi} \| \ ,
 \end{align*}
 which implies~\eqref{eq:Tleqid}.

  We now determine the overlap between $\ket{\Psi}$ and $\ket{\Psi'}$,  \begin{align*}
    \spr{\Psi}{\Psi'}
   & =
    \bra{\Psi}  T_{A B} \otimes \id_R \ket{\Psi} \\
   & =
    \tr(\proj{\Psi} T_{A B} \otimes \id_R)
   =
    \tr(\rho_{A B} T_{A B}) \ .
 \end{align*}
 Because $\rho_{A B}$ has trace one, we have
\begin{align*}
    1 - |\spr{\Psi}{\Psi'}|
  & \leq
    1- \Re \spr{\Psi}{\Psi'}
  =
    \tr\bigl(\rho_{A B} (\id_{A B} - \bar{T}_{A B}) \bigr) \\
   & \leq
     \tr\bigl(\beta_{A B}  (\id_{A B} - \bar{T}_{A B})\bigr) \\
   & =
     \tr(\beta_{A B}) - \tr(\alpha_{A B}^{\half} \beta_{A B}^{\half}) \\
   & \leq
     \tr(\beta_{A B}) - \tr(\alpha_{A B})
   =
     \tr(\Delta_{A B}) \ .
  \end{align*}
  Here, the second inequality follows from the fact that, because of~\eqref{eq:Tleqid}, the operator $\id_{AB} - \bar{T}_{A B}$ is positive and $\rho_{A B} \leq \beta_{A B}$. The last inequality holds because $\alpha_{A B}^{\half} \leq \beta_{A B}^{\half}$, which is a consequence of the operator monotonicity of the square root (Proposition V.1.8 of \cite{bhatia}).

Using \reff{fidelity} and the fact that the fidelity between two pure states is given by their overlap, we find
\begin{align*}
  \| \proj{\Psi} - \proj{\Psi'} \|_1
& \leq
  2 \sqrt{2(1-| \spr{\Psi}{\Psi'} |)} \\
& \leq
  2 \sqrt{2 \tr(\Delta_{A B})}
\leq
  \eps \ .
\end{align*}
Inequality~\eqref{eq:distbound} then follows because the trace distance can only decrease when taking the partial trace.
\end{proof}

\begin{lemma} \label{lem:Hminepsprojbound}
  Let $\rho_{A B}$ and $\sigma_B$ be density operators. Then
  \[
    \Hmin^{\eps}(\rho_{A B}|\sigma_B) \geq \lambda
  \]
  for any $\lambda \in \bbR$ and
  \[
    \eps = \sqrt{8 \tr\bigl(\{\rho_{A B} > 2^{-\lambda} \cdot \id_A \otimes \sigma_B \} \rho_{A B} \bigr)} \ .
  \]
\end{lemma}

\begin{proof}
  Let $\Delta^+_{A B}$ and $\Delta^-_{A B}$ be mutually orthogonal positive operators such that
  \[
    \Delta^+_{A B} - \Delta^-_{A B} = \rho_{A B} - 2^{-\lambda} \cdot \id_A \otimes \sigma_B \ .
  \]
  Furthermore, let $P_{A B}$ be the projector onto the support of $\Delta^+_{A B}$, i.e.,
  \[
    P_{A B} = \{\rho_{A B} > 2^{-\lambda} \cdot \id_A \otimes \sigma_B \} \ .
  \]
  We then have
  \begin{align*}
    P_{A B} \rho_{A B} P_{A B}
  & =
    P_{A B} (2^{-\lambda} \cdot \id_A \otimes \sigma_B + \Delta^+_{A B} - \Delta^-_{A B}) P_{A B} \\
  \geq
    \Delta^{+}_{A B}
  \end{align*}
  and, hence,
  \[
    \sqrt{8 \tr(\Delta^{+}_{A B})}
  \leq
    \sqrt{8 \tr(P_{A B} \rho_{A B})} = \eps \ .
  \]
  The assertion now follows from Lemma~\ref{lem:Hminepsaddbound} because
  \[
    \rho_{A B} \leq 2^{-\lambda} \cdot \id_A \otimes \sigma_B + \Delta^+_{A B} \ .
  \]
\end{proof}

In the following sections we state and prove the relations between the
conditional spectral entropy rates and the smooth conditional max- and
min-entropy.

\subsection{Relation between $\overline{S}(A|B)$ and $H_{\max}^\eps(\rho_{A B} | \rho_B)$}

\begin{theorem}
\label{thm_condmax}
Given a sequence of bipartite states $\hrho^{AB}=\{\rho_n^{AB}\}_{n=1}^\infty$, where
$\rho_n^{AB}\in {\cal{B}}\bigl(({\cal{H}}_A \otimes {\cal{H}}_B)^{\otimes n} \bigr)$,
the sup-spectral conditional entropy rate $\oS(A|B)$, defined by \reff{oscond}, satisfies
\be
\oS(A|B)= \lim_{\eps \rightarrow 0} \limsup_{n \rightarrow \infty}
\frac{1}{n} H^\eps_{\max}(\rho_n^{AB}|\rho_n^{B}),
\label{main4}
\ee
where $ H^\eps_{\max}(\rho_n^{AB}|\rho_n^{B})$ is the smooth max-entropy of the state
$\rho_n^{AB}$ of the sequence, conditional on the corresponding reduced state $\rho_n^B$.
\end{theorem}
\begin{proof}
From the definition \reff{oscond} of $\oS(A|B)$ it follows that for any $\gamma \ge \oS(A|B)$ and any $\delta >0$, for $n$ large enough
\be\tr \big[P_n^\gamma  \rho_n^{AB}\big] > 1 - \delta,\label{one3} \ee
where $P_n^\gamma$ is defined by \reff{proj21}.

 For {{any}} given $\alpha >0$, choose $\gamma= \oS(A|B) + \alpha$, and let
\be {{\rho}}_{n,\gamma}^{AB} := P_n^\gamma  \rho_n^{AB} P_n^\gamma
\label{otrho2}
\ee
Then using (\ref{one3})
and Lemma \ref{gm} we infer that, for $n$ large enough,
${{\rho}}_{n,\gamma}^{AB} \in B^\eps(\rho_n^{AB})$ with
$\eps = 2{\sqrt{\delta}}$. Let ${\pi}_{n,\gamma}^{AB} $ denote the projection
onto the support of ${{\rho}}_{n,\gamma}^{AB} $.

We first prove bound
\be
\lim_{\eps \rightarrow 0} \limsup_{n \rightarrow \infty}
\frac{1}{n} H^\eps_{\max}(\rho_n)\le \oS(A|B).
\label{ofirst3}
\ee

For $n$ large enough,
\bea
H_{\max}^\eps( \rho_n^{AB}|\rho_n^B)
&:=& \inf_{\orho_n\in B^\eps(\rho_n^{AB})}
H_{\max}(\orho_n^{AB}|\rho_n^B)\nonumber\\
&\le & H_{\max}({{\rho}}_{n,\gamma}^{AB}|\rho_n^B)\nonumber\\
&=& \log \tr \bigl((I_n^A \otimes \rho_n^{B}){{\pi}}_{n,\gamma}^{AB} \bigr) \nonumber\\
&\le & \log \tr \bigl((I_n^A \otimes \rho_n^{B})P_n^{\gamma}\bigr)\nonumber\\
&\le & n\gamma
\label{oconverse3}
\eea
The last inequality in \reff{oconverse3} follows from Lemma \ref{lem2}. Hence, for $n$ large
enough,
\be
\frac{1}{n} H_{\max}^\eps( \rho_n^{AB}|\rho_n^B)
\le \gamma = \oS(A|B) + \alpha,
\ee
and since $\alpha$ is arbitrary, we obtain the desired bound \reff{ofirst3}.

To complete the proof of Theorem \ref{thm_condmax}, we assume that
\be
\lim_{\eps \rightarrow 0} \limsup_{n \rightarrow \infty}
\frac{1}{n} H^\eps_{\max}(\rho_n^{AB}|\rho_n^B)< \oS(A|B),
\label{ass2}
\ee
and prove that this leads to a contradiction.
Let $\sigma_{n,\eps}^{AB}$ be the operator for which
\be
H_{\max} (\sigma_{n,\eps}^{AB}|\rho_n^{B}) = \inf_{\orho^{AB} \in B^\eps(\rho_n^{AB})}H_{\max} (\orho^{AB}|\rho_n^{B}).\ee
Hence,
\bea H^\eps_{\max}(\rho_n^{AB}|\rho_n^{B})&=&
H_{\max} (\sigma_{n,\eps}^{AB}|\rho_n^{B})\nonumber\\
&=& \log \tr\bigl((I_n^A \otimes \rho_n^B) \pi_{n,\eps}^{AB} \bigr),
\eea
where $\pi_{n,\eps}^{AB}$ is the projection onto the support of
$\sigma_{n,\eps}^{AB}$.

Hence, the assumption \reff{ass2} is equivalent to the following assumption:
\be
\lim_{\eps \rightarrow 0} \limsup_{n \rightarrow \infty}
\frac{1}{n} \log \tr \bigl[\pi_{n,\eps}^{AB}(I_n^A \otimes \rho_n^B)\bigr]< \oS(A|B).
\label{ass3}
\ee

Note that
\bea
&&
\tr(\pi_{n,\eps}^{AB}  \rho_n^{AB}) \nonumber\\
&=&
\tr \bigl[ \bigl((\rho_n^{AB} - \sigma_{n,\eps}^{AB})
+ \sigma_{n,\eps}^{AB}\bigr)\pi_{n,\eps}^{AB}\bigr]\nonumber\\
&=& \tr \bigl[(\rho_n^{AB} - \sigma_{n,\eps}^{AB}) \pi_{n,\eps}^{AB}\bigr] + \tr \sigma_{n,\eps}^{AB}\nonumber\\
&\ge& \tr \bigl[ \{\rho_n^{AB} \le \sigma_{n,\eps}^{AB}\} (\rho_n^{AB} - \sigma_{n,\eps}^{AB})\bigr] + \tr \bigl[\sigma_{n,\eps}^{AB}\bigr]\nonumber\\
&\ge & - \eps + 1 - \eps = 1- 2\eps.
\label{new1}
\eea
We arrive at the second last line of \reff{new1} using Lemma \ref{lem1}. The last line of \reff{new1} is obtained analogously to \reff{upb4},
since $\sigma_{n,\eps}^{AB} \in B^\eps(\rho_n^{AB})$.


Note, however, that \reff{new1} leads to a contradiction.
This can be seen as follows: Let $R$
be a real number satisfying
$$\tr \bigl[\pi_{n,\eps}^{AB}(I_n^A \otimes \rho_n^B)\bigr] = 2^{n R}.$$
It follows from the assumption \reff{ass3} that, for $\eps$ small enough, $R < \oS(A|B)$. Note that
\bea
&&
\tr(\pi_{n,\eps}^{AB} \rho_n^{AB}) \nonumber \\
&=& \tr\bigl[\pi_{n,\eps}^{AB} (\rho_n^{AB} - 2^{-n\gamma} I_n^A \otimes \rho_n^B)\bigr] \nonumber\\
& &  + 2^{-n\gamma} \tr \bigl[\pi_{n,\eps}^{AB}(I_n^A \otimes \rho_n^B)\bigr]\nonumber\\
&\le & \tr\bigl[\{\rho_n^{AB} \ge 2^{-n\gamma} I_n^A \otimes \rho_n^B\}(\rho_n^{AB} - 2^{-n\gamma} I_n^A \otimes \rho_n^B) \bigr] \nonumber\\
& & + 2^{-n(\gamma - R)}
\nonumber\\
\label{upb7}
\eea
Choose $\oS(A|B) > \gamma > R$. For such a choice, the second term on the right  hand side of
\reff{upb7} tends to zero asymptotically in $n$. However, the first term does not tend to $1$
and we hence obtain the bound
\be
\tr(\pi_{n,\eps}^{AB} \rho_n^{AB}) < 1 - c_0,
\label{compress2}
\ee
for some constant $c_0>0$. This contradicts \reff{new1} in the limit $\eps \rightarrow 0$.
\end{proof}

\subsection{Relation between $\underline{S}(A|B)$ and$H_{\min}^\eps(\rho_{A B} | \rho_B)$}

\begin{theorem}
\label{thm_condmin}
Given a sequence of bipartite states $\hrho^{AB}=\{\rho_n^{AB}\}_{n=1}^\infty$, where
$\rho_n^{AB}\in {\cal{B}}\bigl(({\cal{H}}_A \otimes {\cal{H}}_B)^{\otimes n} \bigr)$,
the inf-spectral conditional entropy rate $\uS(A|B)$
is related to the smooth conditional min-entropy as follows:
\be
\uS(A|B)= \lim_{\eps \rightarrow 0} \liminf_{n \rightarrow \infty}
\frac{1}{n} H^\eps_{\min}(\rho_n^{AB}|\rho_n^{B})
\label{main2}
\ee
\end{theorem}

\begin{proof}

We first prove the bound
\be
\uS(A|B) \ge \lim_{\eps \rightarrow 0} \liminf_{n \rightarrow \infty}
\frac{1}{n} H^\eps_{\min}(\rho_n^{AB}|\rho_n^{B})
\label{second4}
\ee

Let $\sigma_{n,\eps}^{AB}$ be the operator for which
\be
H_{\min} (\sigma_{n,\eps}^{AB}|\rho_n^{B}) = \max_{\orho^{AB} \in B^\eps(\rho_n^{AB})}H_{\min} (\orho^{AB}|\rho_n^{B}).
\label{orhocond}
\ee
Let us define
\bea
&&
\underline{\Upsilon}^{\eps}(A|B)\nonumber\\
&:=& \sup \Big\{ \alpha : \limsup_{n\rightarrow \infty} \mathrm{Tr}\big[
\{\sigma_{n,\eps}^{AB} \ge 2^{-n\alpha} I_n^A \otimes \rho_n^B\}
\Pi_n^\alpha\big] = 0 \Big\},\nonumber\\
\label{ue}
\eea
where $\Pi_n^\alpha:= \sigma_{n,\eps}^{AB}- 2^{-n\alpha} I_n^A \otimes \rho_n^B$.

According to Definition~\ref{def:smoothentropies} of the conditional smooth min-entropy, that to prove (\ref{second4}), it suffices to prove the
following lemma:
\begin{lemma}
\label{above4}
For any sequence of bipartite states $\hrho^{AB} = \{\rho_n^{AB}\}_{n=1}^\infty$, and any $\eps >0$,
there exists an $n_0 \in \mathbb{N}$, such that for all $n \ge n_0$
\be
\uS(A|B) \ge -\frac{1}{n} \log \bigl[ \{\min \{\lambda :
\sigma_{n,\eps}^{AB} \le \lambda I_n^A \otimes \rho_n^B\}\}\bigr],
\label{last4}
\ee
with $\sigma_{n,\eps}^{AB}$ defined by \reff{orhocond}.
\end{lemma}

\begin{proof}
We prove this lemma in two steps. We first prove that for any $\eps >0$ and $n$ large enough,
\be
\underline{\Upsilon}^{\eps}(A|B) \ge-\frac{1}{n} \log\bigl[\min \{\lambda : \sigma_{n,\eps}^{AB} \le \lambda I_n^A \otimes \rho_n^B\}\bigr].
\label{step14}
\ee
We then prove that
\be
\uS(A|B) \ge \lim_{\eps \rightarrow 0} \underline{\Upsilon}^{\eps}(A|B).
\label{step24}
\ee
{\em{Proof of \reff{step14}:}} For any arbitrary $\eta >0$, let $\alpha$ be defined through
the relation
\be
\min \{\lambda:  \sigma_{n,\eps}^{AB} \le \lambda I_n^A \otimes \rho_n^B\} = 2^{- n (\alpha + \eta)}.
\label{124}
\ee
Hence,
\be
-\frac{1}{n} \log \bigl[ \min \{\lambda:  \sigma_{n,\eps}^{AB} \le \lambda I_n^A \otimes \rho_n^B\} \bigr] = \alpha + \eta
\label{alpha14}
\ee
Note that \reff{124} implies that $\sigma_{n,\eps}^{AB} \le  2^{- n (\alpha + \eta)} (I_n^A \otimes \rho_n^B),$ and
hence $(\sigma_{n,\eps}^{AB} - 2^{- n (\alpha + \eta)} I_n^A \otimes \rho_n^B) \le 0 $. This in turn implies
that $(\sigma_{n,\eps}^{AB} - 2^{- n \alpha } I_n^A \otimes \rho_n^B) \le 0 $ and hence
\be
\tr\bigl[ \{\sigma_{n,\eps}^{AB} \ge 2^{- n \alpha } I_n^A \otimes \rho_n^B\} (\sigma_{n,\eps}^{AB} - 2^{- n \alpha } I_n^A \otimes \rho_n^B)\bigr]=0.
\ee
It then follows from the definition \reff{ue} of $\underline{\Upsilon}^{\eps}(A|B)$ that $\alpha \le \underline{\Upsilon}^{\eps}(A|B)$. Hence, using \reff{alpha14},
we get
\be
-\frac{1}{n} \log \bigl[ \min \{\lambda : \sigma_{n,\eps}^{AB} \le \lambda I_n^A \otimes \rho_n^B \}\bigr]- \eta \le\underline{\Upsilon}^{\eps}(A|B),
\ee
which in turn yields \reff{step14}, since $\eta$ is arbitrary.
\smallskip

\noindent
{\em{Proof of \reff{step24}:}} Defining $P_n^\gamma := \{\rho_n^{AB} \ge 2^{-n\gamma}
I_n^A \otimes \rho_n^B\}$, note that
\bea
&&
\tr\bigl[ P_n^\gamma \rho_n^{AB}\bigr]\nonumber\\
&=& \tr\bigl[ P_n^\gamma\sigma_{n,\eps}^{AB}\bigr]
+  \tr\bigl[ P_n^\gamma( \rho_n^{AB}- \sigma_{n,\eps}^{AB})\bigr]
\nonumber\\
&\le&  \tr\bigl[P_n^\gamma (\sigma_{n,\eps}^{AB}
- 2^{-n\alpha} (I_n^A \otimes \rho_n^B)\bigr]\nonumber\\
&& +2^{-n\alpha} \tr\bigl[ P_n^\gamma (I_n^A \otimes \rho_n^B)\bigr] +
\eps \nonumber\\
&\le & \tr\bigl[ \{ \sigma_{n,\eps}^{AB} \ge 2^{-n\alpha} I_n^A \otimes \rho_n^B\} (\sigma_{n,\eps}^{AB}
- 2^{-n\alpha} (I_n^A \otimes \rho_n^B)\bigr]\nonumber\\
&& + 2^{-n(\alpha - \gamma)} + \eps
\label{long4}
\eea
In the above we have made use of Lemma \ref{lem1}, Lemma \ref{lem2} and Corollary \ref{cor1}.

Let us choose $\gamma = \alpha - \delta/2$, for an arbitrary $\delta>0$, with
$\alpha= \underline{\Upsilon}^{\eps}(A|B) - \delta/2$. Then both the first and second terms on the right hand side
of \reff{long4} goes to zero as $n \tends \infty$. Therefore, for $n$ large enough and any $\delta^{'}>0$, in the limit $\eps \tends 0$,
we must have that
\be\tr(P_n^\gamma \rho_n^{AB}) \le \delta^{'},\ee
which in turn implies that $\gamma \le \uS(A|B)$. Hence, from the choice of the parameters $\alpha$ and $\gamma$ it follows that
\be
\lim_{\eps \tends 0} \underline{{\Upsilon}}^{\eps}(A|B) - \delta \le \underline{{{S}}}(A|B),\ee
and since $\delta$ is arbitrary, we obtain the inequality \reff{step24}.
\end{proof}

We next prove the bound
\be
\uS(A|B) \le \lim_{\eps \rightarrow 0} \liminf_{n \rightarrow \infty}
\frac{1}{n} H^\eps_{\min}(\rho_n^{AB}|\rho_n^{B})
\label{second6}
\ee
{\em{Proof of \reff{second6}:}}
Let $\delta > 0$ be arbitrary but fixed. Then by the definition of the inf-spectral conditional entropy rate there exists $\gamma \in {\mathbb{R}}$ such that
\begin{equation}
  \gamma > \uS(A|B) - \delta
\end{equation}
and
\begin{equation}
  \limsup_{n \to \infty} \tr\bigl[\{\rho^{A B}_n \geq 2^{-n \gamma}  \id_n^A \otimes \rho_n^B\} \rho_n^{A B} \bigr] = 0 \ .
\end{equation}
In particular, for any $\eps > 0$ there exists $n_0 \in \bbN$ such that for all $n \geq n_0$.
\begin{multline}
  \tr\bigl[\{\rho^{A B}_n > 2^{-n \gamma} \cdot \id_n^A \otimes \rho_n^B\} \rho_n^{A B} \bigr] \\
\leq
  \tr\bigl[\{\rho^{A B}_n \geq 2^{-n \gamma} \cdot \id_n^A \otimes \rho_n^B\} \rho_n^{A B} \bigr]
<
  \frac{\eps^2}{8} \ .
\end{multline}

Using Lemma~\ref{lem:Hminepsprojbound} we then infer that for all $n \geq n_0$
\begin{equation}
  \Hmin^\eps(\rho_n^{A B} | \rho_n^B)
\geq
  n \gamma
\end{equation}
and, hence
\begin{equation}
  \liminf_{n \to \infty} \frac{1}{n} \Hmin^\eps(\rho_n^{A B} | \rho_n^B)
\geq
  \gamma \ .
\end{equation}
Because this holds for any $\eps > 0$, we conclude
\begin{equation}
  \lim_{\eps \to 0} \liminf_{n \to \infty} \frac{1}{n} \Hmin^\eps(\rho_n^{A B} | \rho_n^B)
\geq
  \gamma > \uS(A|B) - \delta \ .
\end{equation}
The assertion \reff{second6} then follows because this holds for any $\delta > 0$.
\end{proof}

\section{Conclusions}

So far, the \emph{information spectrum approach} and the \emph{smooth
  entropy framework} have been applied within pretty different
subfields of information theory~\footnote{In fact, the two approaches
  have been developed independently within different research
  communities.}.  In the quantum regime, spectral entropy rates have
mostly been used to characterize information sources, communication
channels and entanglement manipulations. In contrast, smooth entropies
proved useful in the context of randomness extraction and
cryptography. We hope that our result bridges the gap between these
two subfields.  In fact, for the study of asymptotic settings where the underlying
resources are available many times, both the information-spectrum
approach and the smooth entropy framework can be used equivalently.

\section{Acknowledgments} The authors are very grateful to Patrick
Hayden for stimulating exchanges, helpful comments and for carefully
reading the proofs. They also thank Masahito Hayashi, Robert K\"onig,
Jonathan Oppenheim and Andreas Winter for interesting discussions. ND
acknowledges the kind hospitality of McGill University
(Montreal), where part of this work was done.

\section{APPENDIX}

In this Appendix we give the proofs of Proposition \ref{equiv_sup} and Proposition \ref{equiv_inf}.

\section{{{{\em{\bf{Proof of Proposition 1}}}}}}

\begin{proof}
For any $\alpha = \overline{\mathcal{D}}(\rho\| \omega) + \delta$, with $\delta
> 0$, implies
\begin{align}
0 &= \lim_{n\rightarrow \infty} \mathrm{Tr}\big[ \{ \rho_n \geq
e^{n\alpha}\omega_n \} \rho_n \big] \nonumber \\
&\geq \lim_{n\rightarrow \infty} \mathrm{Tr}\big[ \{ \rho_n \geq
e^{n\alpha}\omega_n \} (\rho_n - e^{n\alpha}\omega_n) \big] \nonumber \\
&\geq 0
\end{align}
giving $\overline{\mathcal{D}}(\rho\| \omega) \geq \overline{D}(\rho\|
\omega)$, as $\delta$ is arbitrary.  For the converse we assume that the
inequality is strict, such that $\overline{\mathcal{D}}(\rho\| \omega) =
\overline{D}(\rho\| \omega) + 4\delta$ for some $\delta > 0$.  Then choosing
$\alpha = \overline{D}(\rho\| \omega) + 2\delta$, $\gamma = \overline{D}(\rho\|
\omega) + \delta$, we have from Lemma \ref{lem1},
\begin{align}
\mathrm{Tr}\big[ \{ \rho_n \geq e^{n\alpha}\omega_n \} \rho_n \big] &\leq
\mathrm{Tr}\big[ \{ \rho_n \geq e^{n\gamma}\omega_n \} (\rho_n -
e^{n\gamma}\omega_n) \big] \nonumber \\
&\phantom{=}\: + e^{n\gamma}\mathrm{Tr}\big[ \{ \rho_n \geq e^{n\alpha}\omega_n
\} \omega_n \big] \nonumber \\
&\leq \eps_n + e^{-n\delta}
\end{align}
where $\eps_n = \mathrm{Tr}\big[ \{ \rho_n \geq e^{n\gamma}\omega_n \}
(\rho_n - e^{n\gamma}\omega_n) \big]$ and $\mathrm{Tr}\big[ \{ \rho_n \geq
e^{n\alpha}\omega_n \}\omega_n \big] \leq e^{-n\alpha}$ holds for any $\alpha$.
As the right hand side goes to zero asymptotically and since $\alpha <
\overline{\mathcal{D}}(\rho\| \omega)$ we have a contradiction.
\end{proof}

\section{{{{\em{\bf{Proof of Proposition 2}}}}}}

\begin{proof}
For any $\alpha = \underline{D}(\rho\| \omega) - \delta$, with $\delta > 0$,
implies
\begin{align}
1 &\geq \lim_{n\rightarrow \infty} \mathrm{Tr}\big[ \{ \rho_n \geq
e^{n\alpha}\omega_n \} \rho_n \big] \nonumber \\
&\geq \lim_{n\rightarrow \infty} \mathrm{Tr}\big[ \{ \rho_n \geq
e^{n\alpha}\omega_n \} (\rho_n - e^{n\alpha}\omega_n) \big] \nonumber \\
&= 1
\end{align}
giving $\underline{\mathcal{D}}(\rho\| \omega) \geq \underline{D}(\rho\|
\omega)$, as $\delta$ is arbitrary.  For the converse we assume that the
inequality is strict, such that $\underline{\mathcal{D}}(\rho\| \omega) =
\underline{D}(\rho\| \omega) + 4\delta$ for some $\delta > 0$.  Then choosing
$\alpha = \underline{\mathcal{D}}(\rho\| \omega) - \delta$, $\gamma =
\underline{\mathcal{D}}(\rho\| \omega) - 2\delta$, we have from Lemma
\ref{lem1},
\begin{align}
1 &\overset{n\rightarrow \infty}\leftarrow \mathrm{Tr}\big[ \{ \rho_n \geq
e^{n\alpha}\omega_n \} \rho_n \big] \nonumber \\
&\leq \mathrm{Tr}\big[ \{ \rho_n \geq e^{n\gamma}\omega_n \} (\rho_n -
e^{n\gamma}\omega_n) \big] \nonumber \\
&\phantom{=}\: + e^{n\gamma}\mathrm{Tr}\big[ \{ \rho_n \geq e^{n\alpha}\omega_n
\} \omega_n \big] \nonumber \\
&\leq \mathrm{Tr}\big[ \{ \rho_n \geq e^{n\gamma}\omega_n \} (\rho_n -
e^{n\gamma}\omega_n) \big] + e^{-n\delta}
\end{align}
where $\mathrm{Tr}\big[ \{ \rho_n \geq e^{n\alpha}\omega_n \}\omega_n \big]
\leq e^{-n\alpha}$ holds for any $\alpha$. Thus $\lim_{n\rightarrow \infty}
\mathrm{Tr}\big[ \{ \rho_n \geq e^{n\gamma}\omega_n \} (\rho_n -
e^{n\gamma}\omega_n) \big] = 1$, where $\gamma > \underline{D}(\rho\|\omega)$,
which is a contradiction.
\end{proof}


\begin{thebibliography}{1}


\bibitem{BBCM95}
C.~H. Bennett, G.~Brassard, C.~Cr{\'e}peau, and U.~Maurer,
``Generalized privacy amplification,''
\emph{IEEE Trans. Inf. Theory}, vol. 41, pp.~1915--1923, 1995.

\bibitem{bhatia}
R. Bhatia, \emph{Matrix Analysis}, Springer.

\bibitem{bd1}
G.~Bowen and N.~Datta,
``Beyond {i.i.d.} in quantum information theory,''
\emph{arXiv:quant-ph/0604013}, {\em{Proceedings of the
2006 IEEE International Symposium
on Information Theory}}, 2006.

\bibitem{bd_rev}
G.~Bowen and N.~Datta,
``Quantum coding theorems for arbitrary sources, channels and entanglement resources,''
\emph{arXiv:quant-ph/0610003}, 2006

\bibitem{bd_entpure}
G.~Bowen and N.~Datta,
`` Asymptotic entanglement manipulation of bipartite pure states,''
\emph{arXiv:quant-ph/0610199}, 2006.

\bibitem{bd_entmixed}
G.~Bowen and N.~Datta,
`` Entanglement cost for sequences of arbitrary quantum states,''
\emph{arXiv:0704.1957}, 2007.

\bibitem{Cachin97}
C.~Cachin,
``Smooth entropy and \mbox{R\'enyi} entropy,''
in {\em Advances in Cryptology --- \mbox{EUROCRYPT '97}}, LNCS, vol. 1233, pp. 193--208. Springer, 1997,

\bibitem{DFRSS07}
I.~Damg{\aa}rd, S.~Fehr, R.~Renner, L.~Salvail, and C.~Schaffner,
``A tight high-order entropic uncertainty relation with applications,''
in {\em Advances in Cryptology --- CRYPTO 2007}, LNCS, vol. 4622, pp. 360--378. Springer, 2007.

\bibitem{DFSS05}
I.~Damgaard, S.~Fehr, L.~Salvail, and C.~Schaffner,
``Cryptography in the bounded quantum-storage model,'' in {\em 46th Annual Symposium on Foundations of Computer Science
  (FOCS)}, pp. 449--458, 2005.

\bibitem{DFSS07}
I.~Damg{\aa}rd, S.~Fehr, L.~Salvail, and C.~Schaffner,
``Secure identification and {QKD} in the bounded-quantum-storage model,''
in {\em Advances in Cryptology --- CRYPTO 2007}, vol.~4622, pp. 342--359. Springer, 2007.

\bibitem{han}
T.~S. Han, \emph{Information-Spectrum Methods in Information Theory}, Springer-Verlag, 2002.

\bibitem{hanverdu93}
T.~S. Han and S.~Verdu, ``Approximation theory of output statistics,'', \emph{IEEE
Trans. Inform. Theory}, vol.~39, pp. 752--772, 1993.

\bibitem{hayashi03}
M.~Hayashi and H.~Nagaoka, ``General formulas for capacity of
  classical--quantum channels,'' \emph{IEEE Trans. Inform. Theory}, vol.~49,
  pp. 1753--1768, 2003.

\bibitem{hay_conc}
M.~Hayashi,
``General formulas for fixed-length quantum entanglement concentration,''\emph{IEEE Trans. Inform. Theory}, Vol.~52,
No. 5, 1904-1921, 2006.

\bibitem{Mat07} K.~Matsumoto, ``Entanglement cost and distillable entanglement of symmetric states,'' \emph{arXiv:0708.3129}, 2007.

\bibitem{nagaoka02}
H.~Nagaoka and M.~Hayashi, ``An information-spectrum approach to classical and
  quantum hypothesis testing for simple hypotheses,'' \emph{arXiv:quant-ph/0206185},
  2002.

\bibitem{nielsen}
M.~A.Nielsen and I.~L.Chuang,
  {\em{Quantum Computation and Quantum Information}},
Cambridge University Press, Cambridge, 2000.

\bibitem{ogawa00}
T.~Ogawa and H.~Nagaoka, ``Strong converse and stein's lemma in quantum
  hypothesis testing,'' \emph{IEEE Trans. Inform. Theory}, vol.~46, pp.
  2428--2433, 2000.

\bibitem{ogawanagaoka02}
T.~Ogawa and H.~Nagaoka,
``New proof of the channel coding theorem via hypothesis testing in
  quantum information theory, '', \emph{arXiv:quant-ph/0208139}, 2002.

\bibitem{renatophd}
R.~Renner,
``Security of quantum key distribution,'' PhD thesis, ETH Zurich, \emph{arXiv:quant-ph/0512258}, 2005.

\bibitem{KR}
R.~Renner and R.~Koenig,
``Universally composable privacy amplification against quantum adversaries,''
\emph{Proc. of TCC 2005, LNCS, Springer}, vol.~3378, 2005.

\bibitem{RenWol04b}
R.~Renner and S.~Wolf,
``Smooth {R\'enyi} entropy and applications,''
in {\em Proc.\ International Symposium on Information Theory}, p.~233. IEEE, 2004.

\bibitem{RenWol05b}
R.~Renner and S.~Wolf,
``Simple and tight bounds for information reconciliation and privacy
  amplification,''
in {\em Advances in Cryptology --- ASIACRYPT 2005}, LNCS, vol.~3788, pp. 199--216. Springer, 2005.

\bibitem{Renyi61}
A.~R{\'e}nyi,
``On measures of entropy and information,''
in {\em Proceedings of the 4th Berkeley Symp.\ on Math.\ Statistics
  and Prob.}, vol.~1, pp 547--561. Univ.\ of Calif.\ Press, 1961.

\bibitem{ReWoWu07}
R.~Renner, S.~Wolf, and J.~Wullschleger,
``Trade-offs in information-theoretic multi-party one-way key
  agreement,'' to appear in Proc.\ of ICITS 2007, 2007.

\bibitem{ScaRen07}
V.~Scarani and R.~Renner, ``Quantum cryptography with finite resources,''
\emph{arXiv:0708.0709}, 2007.

\bibitem{STTV07}
B.~Schoenmakers, J.~Tijoelker, P.~Tuyls, and E.~Verbitskiy,
``Smooth {R\'enyi} entropy of ergodic quantum information sources,''
\emph{arXiv:0704.3504}, 2007.

\bibitem{verdu94}
S.~Verdu and T.~S. Han, ``A general formula for channel capacity,'' \emph{IEEE
  Trans. Inform. Theory}, vol.~40, pp. 1147--1157, 1994.

\bibitem{winter99}
A.~Winter, ``Coding theorem and strong converse for quantum channels,'',  \emph{IEEE Trans. Inf. Theory}, vol.~45, 1999.


\end{thebibliography}
\end{document}